\RequirePackage{fix-cm}
\documentclass[twocolumn]{svjour3}          
\smartqed  
\usepackage{amsmath}

\usepackage{graphicx}

\usepackage{float} 
\usepackage{here}

\usepackage[subrefformat=parens]{subcaption}
\begin{document}

\title{A Second Order Model of Capacity Drop at Expressway Lane-Drop Bottlenecks}

\subtitle{}


\author{Fuminori Hattori         \and
        Kentaro Wada 
}


\institute{Fuminori Hattori \at
Graduate School of Science and Technology
Degree Programs in Systems and Information Engineering, University of Tsukuba, 
Tsukuba, Ibaraki, 305-8753, Japan\\
              \email{s202520481@u.tsukuba.ac.jp}           
           \and
           Kentaro Wada \at
           Institute of System and Information Engineering, University of Tsukuba, 
           Tsukuba, Ibaraki, 305-8753, Japan.\\
	\email{wadaken@sk.tsukuba.ac.jp} (Corresponding Author)
}

\date{Received: date / Accepted: date}

\maketitle

\begin{abstract}
This paper presents a second-order model of capacity drop at expressway lane-drop bottlenecks.
The model is an extension of Jin's model \cite{Jin2017}. 
This model captures not only the stationary state associated with the capacity drop but also the transitional dynamics leading from the onset of congestion to that state.
The characteristics of the proposed model are examined theoretically and numerically.
The results show that the capacity drop stationary state is stable and is reached immediately once congestion occurs.
Furthermore, we validate the model using empirical data. 
 {The results suggest that the model has the potential to provide new insights into congestion phenomena at expressway lane-drop bottlenecks.}
\keywords{Capacity drop \and Lane-drop bottlenecks \and Second order model \and Calibration \and Validation}
\end{abstract}

\section{Introduction}
\label{intro}
On expressways, it has been observed that the discharge flow rate during congestion is approximately 10 \% lower than that before the onset of congestion ~\cite{Chung,Banks,Cassidy1999}. 
This phenomenon is referred to as the capacity drop. 
Since the decrease in discharge flow rate prolongs the duration of congestion and causes additional delays for drivers, it has motivated the development of various traffic control strategies.

There are two major hypotheses regarding the causes of the capacity drop phenomenon, both based on driving behavior. 
One attributes it to disturbances in traffic flow caused by lane-changing maneuvers~\cite{Cassidy2005}, and the other to sluggish acceleration behavior of drivers when exiting bottlenecks~\cite{Hall}. 
Based on the former hypothesis, Leclercq et al.~\cite{Leclercq} proposed a model that endogenously reproduces the capacity drop phenomenon at merging sections. 
Their model is a hybrid of microscopic and macroscopic approaches, and incorporates assumptions from two models. One is the Newell–Daganzo merge model~\cite{Daganzo1995}, which allocates downstream capacity between two merging streams. The other is Laval's model~\cite{Laval}, which assumes that merging vehicles create voids in the traffic stream, thereby reducing the flow. 
The validity of the model has been confirmed by comparison with empirical data. 
However, Cassidy and Rudjanakanoknad~\cite{Cassidy2005} argued that ``lane changing alone might not explain the capacity drop.''

The latter hypothesis suggests that the decrease in flow is caused by low acceleration of vehicles starting from queues. 
A macroscopic model that endogenously captures the capacity drop based on this idea was proposed by Jin~\cite{Jin2017}. 
This is a first-order model that considers lane-drop bottlenecks, such as shown in Fig.~\ref{l}, which appears in a merging area, a lane-drop area, a work zone, or an accident zone.
This model is based on the following two assumptions: (i) spatially inhomogeneous fundamental diagrams (FDs), and (ii) bounded acceleration (BA) of vehicles. 
However, this model assumes that the capacity drop stationary state occurs instantaneously after the onset of congestion, and therefore fails to capture the transition process from the congestion onset to the stationary state. 
In addition, its results have not been validated against sufficient empirical data.

In this study, we extend the model proposed by Jin~\cite{Jin2017} to a second-order model of capacity drop at a lane-drop bottleneck.
The extended model enables us to analyze the theoretical characteristics of the transition process from the onset to the stationary state of the congestion associated with the capacity drop phenomenon (``capacity drop stationary state"). 
Furthermore, we calibrate the model using empirical data and validate whether it can consistently explain real-world traffic phenomena.

The theoretical extension and analysis in this study follow a similar framework to that used by Wada et al.~\cite{Wada}, which extended Jin~\cite{Jin2018}'s model of sag and tunnel bottlenecks. 
However, the applicability of such an extension to lane-drop sections is not obvious. 
Moreover, clarifying the common and distinct characteristics between lane-drop and sag/tunnel sections through this theoretical framework is valuable.

 {
The scope of the present model covers not only lane-drop sections on the mainline but also those at on-ramp merging areas. 
While the two situations are not strictly the same, for example in terms of microscopic vehicle maneuvers, at a coarse scale they are bottlenecks primarily caused by the reduction in the number of lanes. 
It is therefore reasonable to assume that the macroscopic model proposed in this study can describe the capacity drop phenomenon in both types through a common mechanism (see also, for similar treatments, Jin \cite{Jin2017}; Laval and Daganzo \cite{Laval}; Laval et al.\cite{Laval2007}).
}

\begin{figure}[tb]
\centering
\includegraphics[width=0.48\textwidth,clip]{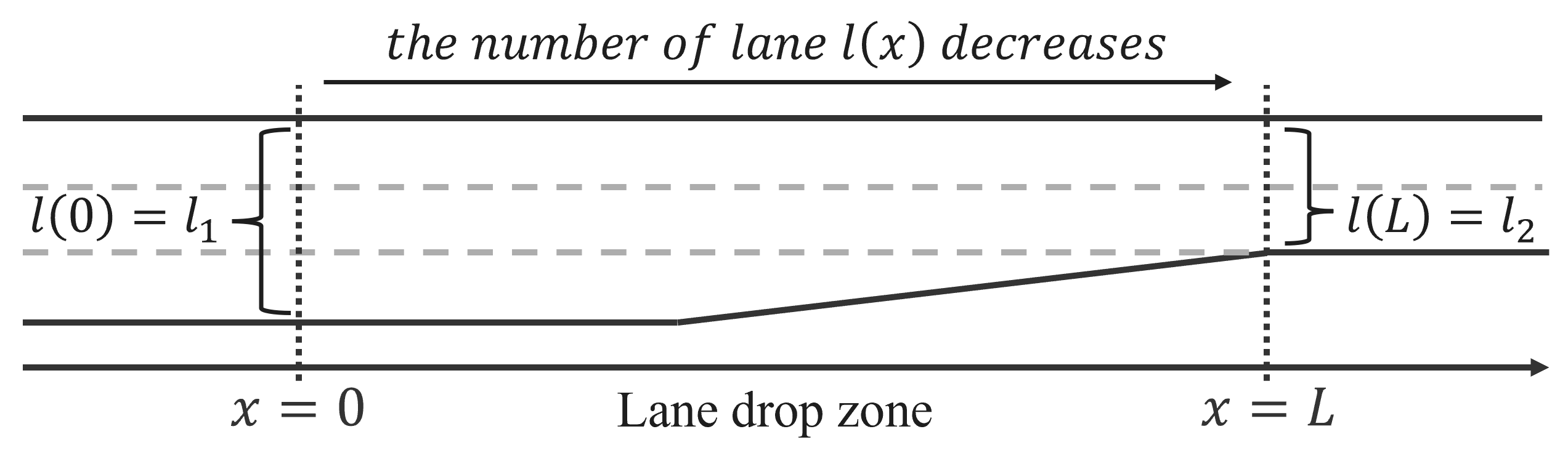}
\caption{Lane drop zone}\label{l}
\end{figure}

The structure of this paper is as follows. 
In Section~\ref{2}, we present the extension of Jin's~\cite{Jin2017} model to a second-order model. 
Simulation results are compared with those of Jin's original model under stationary conditions.
In Section~\ref{sec:Theory}, we conduct theoretical analysis and simulations to demonstrate the characteristics of the capacity drop.
In Section~\ref{sec:empirical}, we validate the extended model using real-world data (Zen Traffic Data~\cite{Zen}).
In Section~\ref{sec:conclusion}, we summarize this study and discuss directions for future work.

\section{Model}
\label{2}

\subsection{Formulation}

\begin{figure}[tb]
\centering
\includegraphics[width=0.3\textwidth,clip]{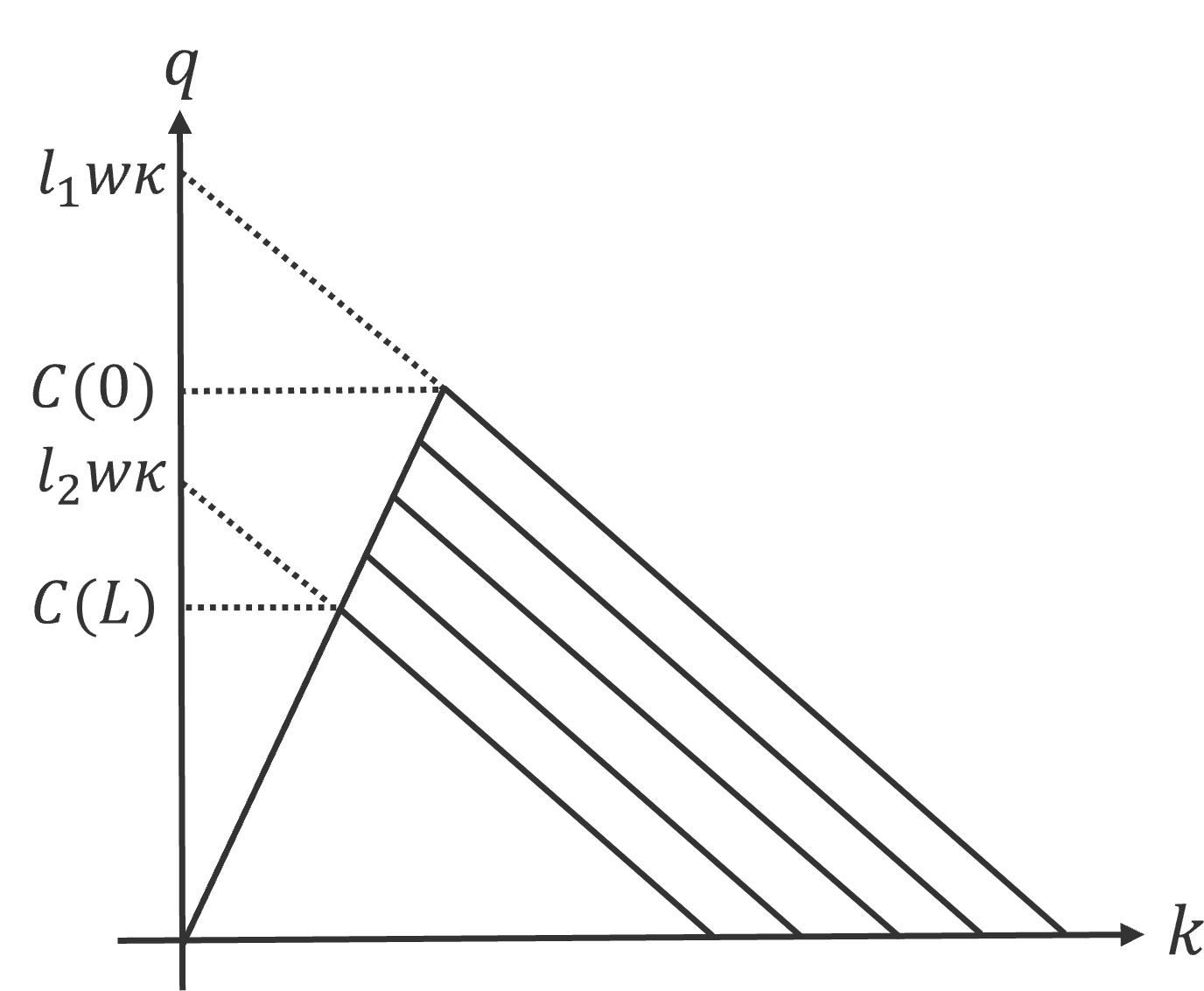}
\caption{Flow-density fundamental diagram}\label{FD}
\end{figure}

In this section, we present a second-order model for a lane-drop bottleneck based on the BA-LWR model, as illustrated in Fig.\ref{l}. 
This model extends the first-order model by Jin~\cite{Jin2017}, and is as follows: 
\begin{subequations}\label{BALWR}
\begin{align}
& k_t + (kv)_x = 0 \label{eq:conservation}, \\
& v_t + v v_x = \min \left\{ A(v,x), \frac{V(k,x) - v}{\epsilon} \right\}\label{eq:velocity},
\end{align}
\end{subequations}
where \( k \) and \( v \) denote traffic density and speed, respectively. \( V(k, x) \) is the speed-density FD, and the traffic flow rate \( q \) is given by \( kv \). The parameter \( \epsilon = \lim_{\Delta t \to 0^+} \Delta t \) is a hyperreal infinitesimal number (see Jin~\cite{Jin2019} for details), and subscripted variables denote partial derivatives. Eq.~\eqref{BALWR} is equivalent to the two-phase model proposed by Lebacque~\cite{Laba}.

This model is based on the following two assumptions.
The first assumption is inhomogeneous fundamental diagram~(FD) that depends on the number of lanes $l(x)$ at each location $x$:
\begin{align}
q = \min \left \{ uk, w\left(l(x)\kappa - k\right) \right \}.\label{sec:FD}
\end{align}
The capacity also varies with location:
\begin{align}
C(x) = \frac{uw}{u+w} l(x) \kappa, 
\label{Cx}
\end{align}
where, \( u \), \( w \), \( \kappa \), and \( l(x)\kappa \) represent the free-flow speed, backward wave speed, jam density of single lane, and jam density of all lanes, respectively. 
More specifically, we assume that \( l(x) \) decreases linearly in the lane-drop bottleneck section, as shown in Fig.~\ref{l},
 {\begin{align}
l(x) = \max\left\{l_{2}, \min\left\{l_{1}, l_{1}-\frac{l_{1} - l_{2}}{L}x\right\} \right\}. \label{lx}
\end{align}}
Based on this assumption, the FD continuously shrinks within the bottleneck section, and the capacity \( C(x) \) decreases continuously from upstream to downstream, as illustrated in Fig.~\ref{FD}.
 {In addition, we can account for the reduction in the effective number of lanes caused by systematic lane changing, as in Jin~\cite{Jin2010}: if the lane-changing intensity is $\eta(x)$ at $x$, the effective number of lanes is given by $\frac{l(x)}{1 + \eta(x)}$.}
In this section, however, we simplify the analysis by assuming \( \eta(x) = 0 \).

The second assumption is that vehicle acceleration is bounded (BA: Bounded Acceleration). This  addreses the issue in the KW theory that allows for infinite acceleration. The bounded acceleration function satisfies: (i) non-negativity: \( A(v, x) \ge 0 \), (ii) boundedness: \( A(v, x) \le a_0 \), and (iii) non-increasing with respect to speed: \( \frac{\mathrm{d}A(v, x)}{\mathrm{d}v} \le 0 \). In Sections \ref{2} and \ref{sec:Theory}, we adopt the constant acceleration model for comparison with Jin's model~\cite{Jin2017}:
\begin{align}
A(v, x) = a_0 - g\phi(x).
\label{AC}
\end{align}
For the validation with empirical data in Section \ref{sec:empirical}, we adopt the TWOPAS model:
\begin{align}
A(v, x) = (a_0 - g\phi(x))\left(1 - \frac{v}{u} \right),
\label{AT}
\end{align}
where \( a_0 \) is the maximum acceleration, \( g = 9.8~\mathrm{m/s^2} \) is the gravitational acceleration, and \( \phi(x) \) denotes the decimal gradient.

\subsection{Continuum Model in Lagrangian Coordinate}
\label{continuum}

\begin{figure}[tb]
\centering
\includegraphics[width=0.48\textwidth,clip]{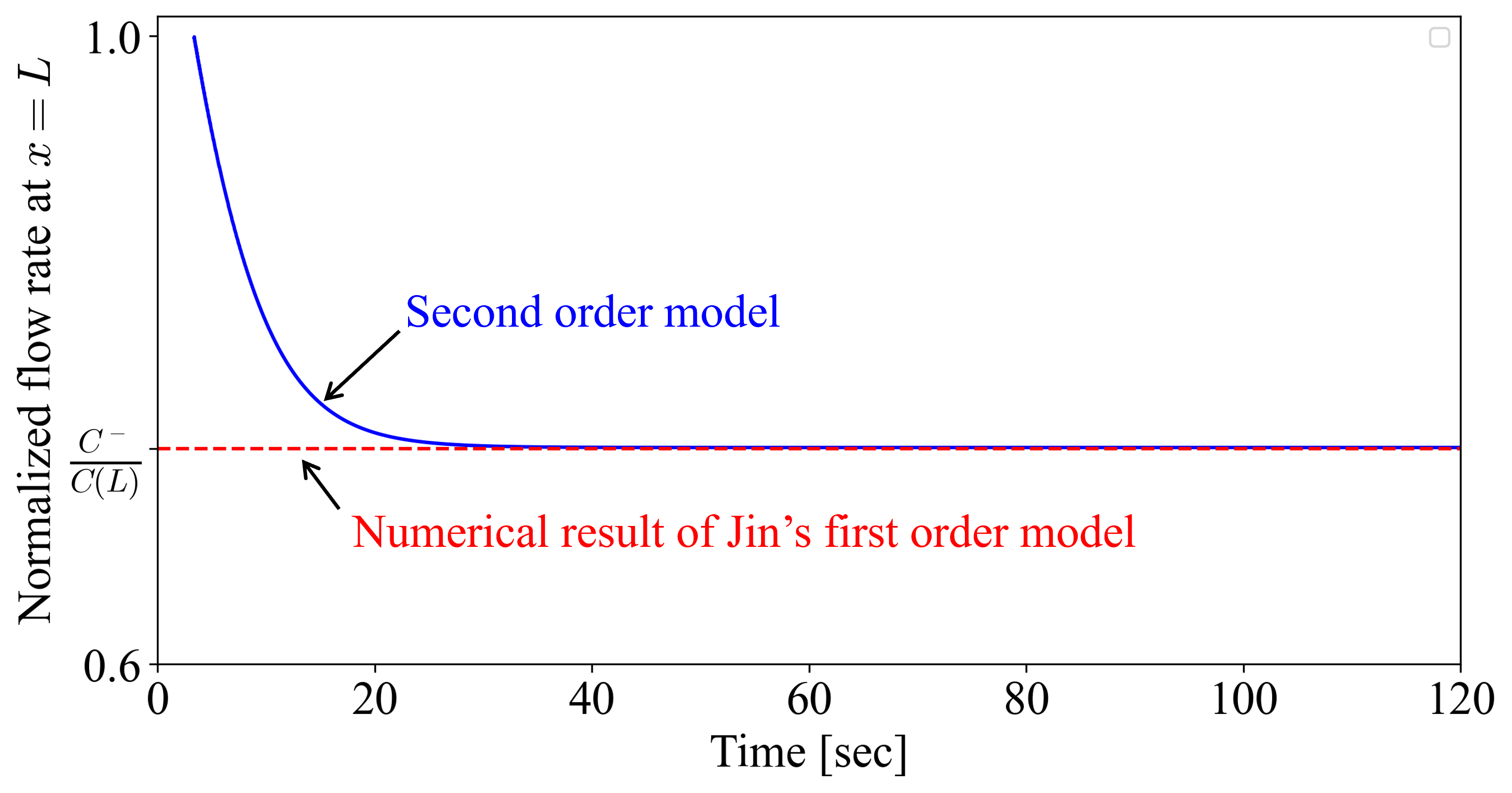}
\caption{Evolution of queue discharge flow rate at $x=L$ under the second order model}
\label{kaiseki}
\end{figure}

\begin{figure}[tb]
\centering
\includegraphics[width=0.48\textwidth,clip]{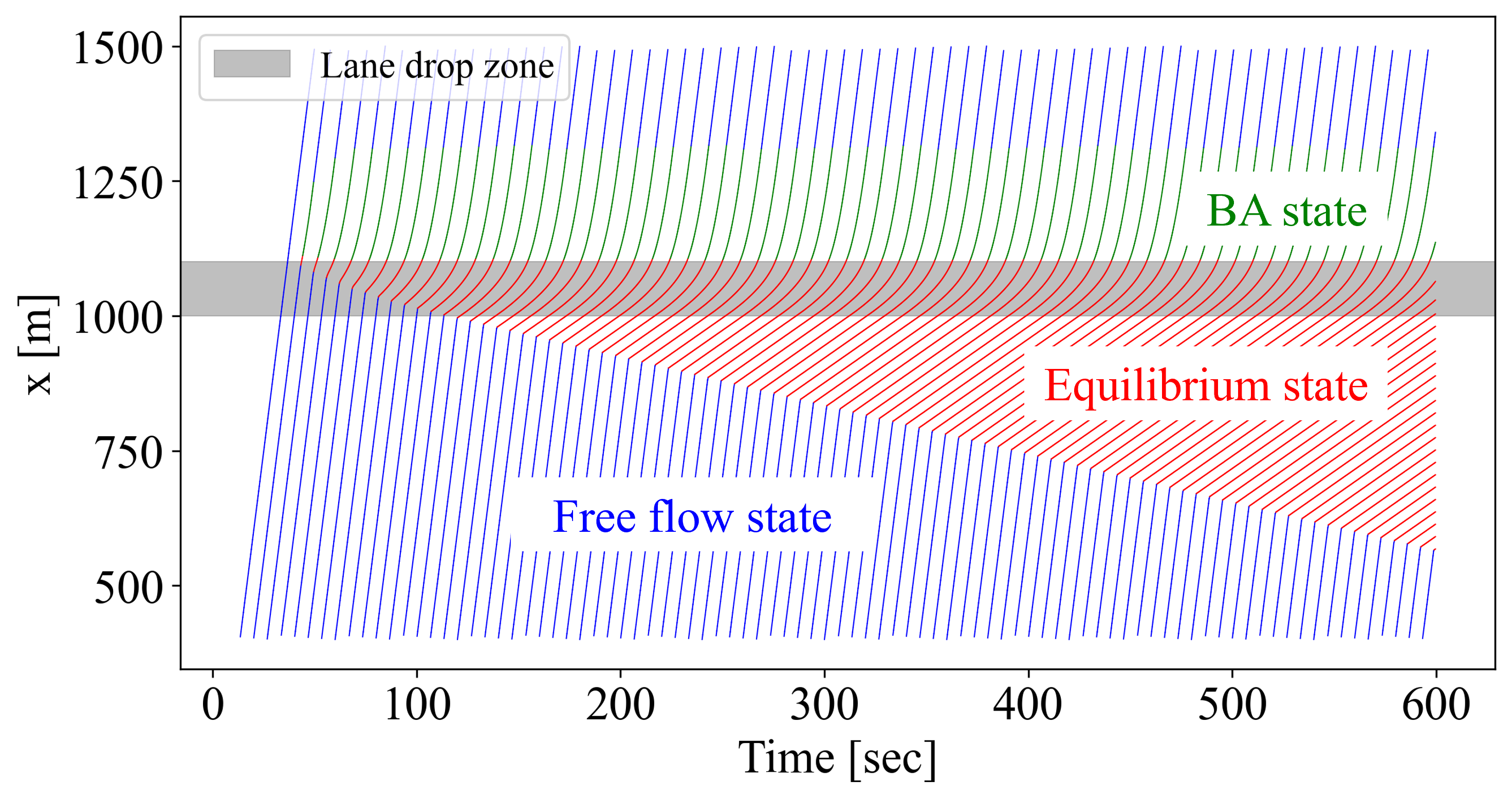}
\caption{The location of the $n$th vehicle at time $t$ (where color indicates traffic states)}
\label{sim}
\end{figure}

\begin{figure}[tb]
\centering
\includegraphics[width=0.48\textwidth,clip]{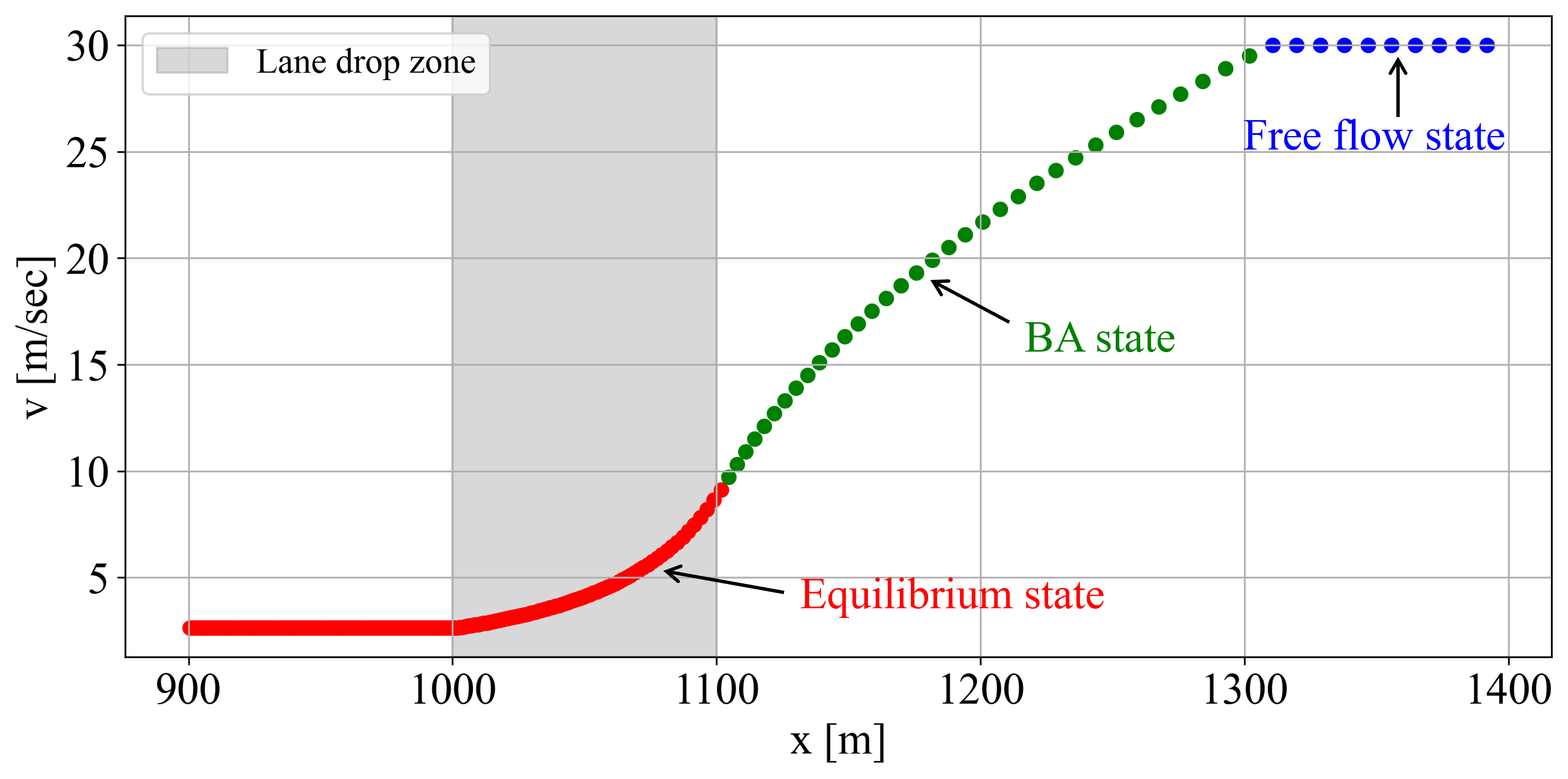}
\caption{Speed recovery profile in the stationary state of capacity drop (where color indicates traffic states)}
\label{spdp}
\end{figure}

In this section, we first transform Eq.~\eqref{BALWR} into the Lagrangian coordinate system using the method developed by Jin~\cite{Jin2016}. 
This transformation enables efficient numerical computation via simple simulation procedures. 
Using this simulation, we verify whether the transition process of the capacity drop phenomenon is well captured and whether the stationary state results are consistent with those in Jin~\cite{Jin2017}.

This transformation consists of the following two steps:  
(i) transformation of state variables, and  
(ii) finite difference approximation of derivatives.
In step (i), vehicle speed \( v(t, n) \) and spacing \( s(t, n) \) are expressed in terms of traffic flow \( q(t, x) \) and density \( k(t, x) \) as follows:
\begin{align}
v(t, n) &\equiv \frac{\partial X(t, n)}{\partial t} = \frac{q(t, x)}{k(t, x)} \label{vtn}, \\
s(t, n) &\equiv -\frac{\partial X(t, n)}{\partial n} = \frac{1}{k(t, x)} \label{stn},
\end{align}
where \( n \) represents the cumulative vehicle number, and the continuous function \( X(t,n) \) indicates the location of the \( n \)th vehicle at time \( t \).
In step (ii), differential operations are approximated using finite differences with time step \( \Delta t \) and vehicle index step \( \Delta n \):
\begin{align}
v(t, n) &\approx \frac{X(t, n)-X(t - \Delta t, n)}{\Delta t} \label{v}, \\
s(t, n) &\approx \frac{X(t, n - \Delta n)-X(t, n)}{\Delta n} \label{s}, \\
a(t, n) &\approx \frac{v(t + \Delta t, n)-v(t, n)}{\Delta t} \label{a}.
\end{align}

By using the above transformations and substituting into Eq.~\eqref{BALWR}, we obtain the following continuum model in Lagrangian coordinates  {(see Appendix \ref{sec:derivation} for the details of the derivation)}:
\begin{subequations}\label{overallCF}
\begin{align}
\begin{split}
X(t + \Delta t, n) &= X(t, n) + \\
&\quad \min \left\{ V(s, x),\ v(t, n) + A(v, x)\, \Delta t \right\} \Delta t,
\end{split} \label{CF} \\
V(s, x) &= \min \left\{ u,\ \frac{s(t, n) - d(X(t, n))}{\tau(X(t, n))} \right\}, \label{CF2} \\
s(t, n) &\equiv \frac{X(t, n - \Delta n) - X(t, n)}{\Delta n}, \label{CF3}
\end{align}
\end{subequations}
where \( d(x) = \frac{1}{l(x)\kappa} \) and \( \tau(x) = \frac{1}{l(x)w\kappa} \). 
Although Eq.~\eqref{overallCF} is usually interpreted as Newell's simplified car-following model~\cite{Newell}, this study treats it strictly as a continuum model. 
This is because the variables in Eq.~\eqref{overallCF} cannot be interpreted directly as vehicle-level variables due to the multi-lane nature of the model. 
In other words, \( d(X(t, n)) \) and \( \tau(X(t, n)) \) serve as variables for simulation and should not be interpreted as minimum spacing or safe time gap as in Wada et al.~\cite{Wada}.

We proceed to simulate the extended model and examine whether it properly describes the transition process and whether the steady-state values agree with those computed by Jin~\cite{Jin2017}. 
For comparison purposes, we use the same parameters as in Section 7 of Jin~\cite{Jin2017}: lane-drop section length \( L = 100~\mathrm{m} \), upstream number of lanes \( l_1 = 2 \), downstream number of lanes \( l_2 = 1 \), free-flow speed \( u = 30~\mathrm{m/s} \), backward wave speed \( w = 5~\mathrm{m/s} \), jam density \( \kappa = 1/7~\mathrm{veh/m} \), maximum acceleration \( a_0 = 2~\mathrm{m/s^2} \). The time step and vehicle index step are set \( \Delta t = 0.006~\mathrm{s} \) and \( \Delta n = 0.01~\mathrm{veh} \), respectively. 
Let \( C^- \) denote the queue discharge flow (QDF) rate. 
Then, the capacity drop ratio can be defined as \( C^-/C(L) \).

The simulation results are presented in Fig.~\ref{kaiseki}. 
The blue line shows the time evolution of the normalized flow rate at the downstream end of the bottleneck \( x = L \), where the flow is normalized by the downstream capacity \( C(L) \). 
From this, we can confirm that the extended theory successfully captures the transition process of the capacity drop phenomenon and that the stationary state value is consistent with that of Jin~\cite{Jin2017}.

Fig.~\ref{sim} shows the traffic states over time \( t \) and space \( x \), and Fig.~\ref{spdp} illustrates the speed profile as a function of position under the capacity drop stationary state. 
In these figures, the blue line represents free-flow conditions, the red line indicates the equilibrium congested state of the inhomogeneous LWR model, the green line denotes acceleration states under the bounded acceleration constraint, and the gray area marks the bottleneck section.
From these results, we observe the following:  
(i) the model exhibits a transition from equilibrium congested state to BA state at the downstream end of the bottleneck (Fig.~\ref{sim});  
(ii) under the capacity drop stationary state, the speed profile exhibits a convex pattern within the bottleneck and a concave pattern downstream of the bottleneck (Fig.~\ref{spdp}).

\section{Theoretical Analysis}\label{sec:Theory}

\subsection{Reduced Model}\label{sec:RM}

\begin{figure}[tb]
\centering
\includegraphics[width=0.48\textwidth,clip]{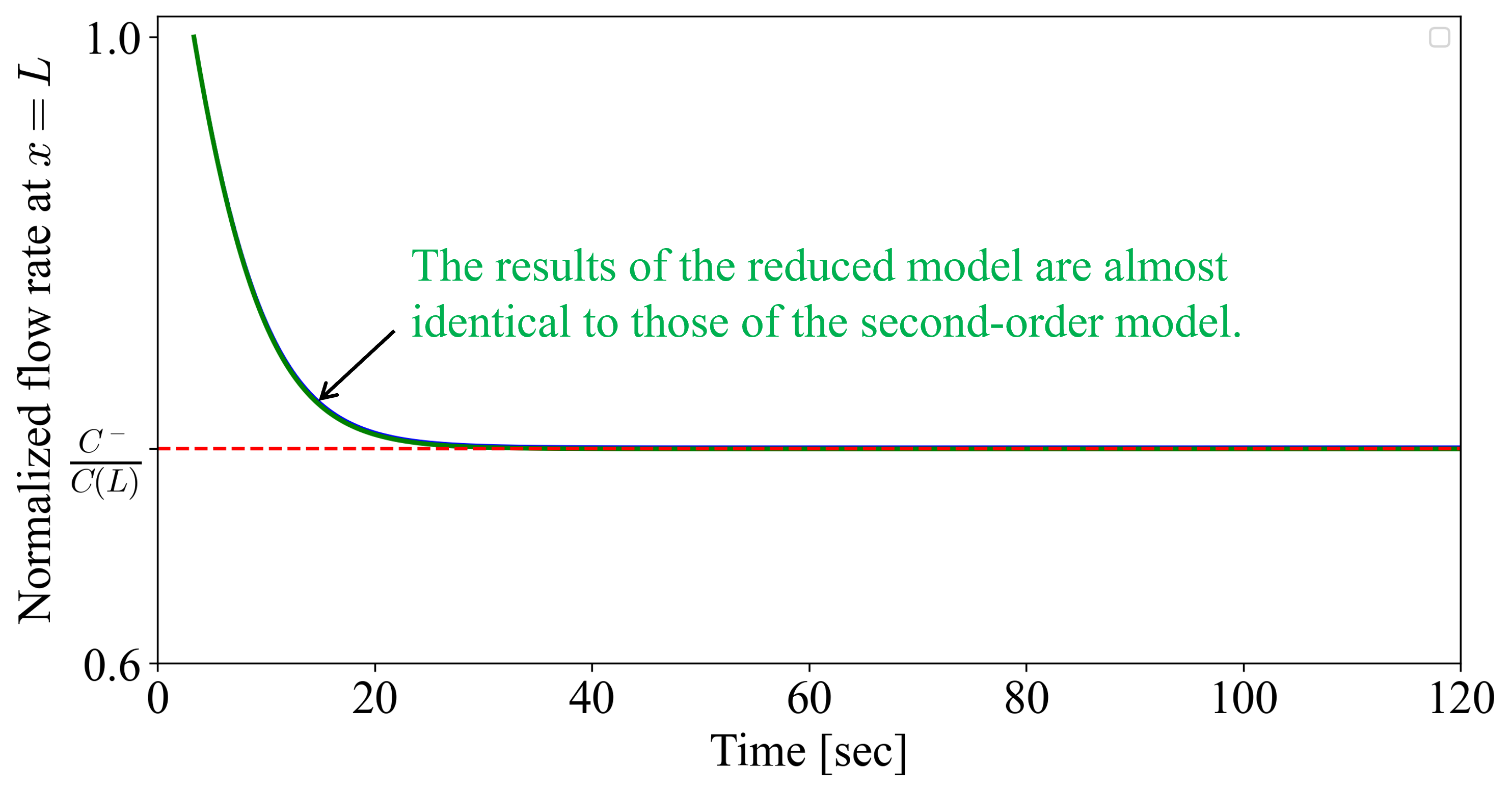}
\caption{Evolution of queue discharge flow rate at $x=L$ (blue line: continuum model in Lagrangian coordinate; green line: reduced model)}
\label{kaisekiIFS}
\end{figure}

In Section~\ref{2}, we demonstrated that the extended model can be easily computed through simulation. However, this approach has some drawbacks: it is difficult to analytically examine the theoretical characteristics, and numerical errors can occur. 
To address these limitations, in this section, we present a one-dimensional, spatially reduced model of capacity drop for lane-drop bottlenecks. 
This is formulated as a mathematical model known as an iterated function system.
This model built on the relationship between two types of traffic states across the downstream end of the bottleneck, as identified in the previous section.

More specifically, this model is constructed in the following two steps:
first, the equilibrium speed  {\( v_{n+\Delta n} \)} at the downstream end of the bottleneck \( x = L \) is computed from the BA speed \( \tilde{v}_{n} \) at $x = L + d(L)\Delta n$:
 {
\begin{align}
v_{n + \Delta n} = \left(\tau_{x}(L)\Delta n + \frac{1 + d_{x}(L)\Delta n}{\tilde{v}_{n}}\right)^{-1}
\label {EQBA}
\end{align}}
This is derived from the relationship of the equilibrium congested state (see Wada et al.~\cite{Wada} for details).

Next, the BA vehicle speed \( \tilde{v}_{n} \) at $x = L + d(L)\Delta n$ is computed from the equilibrium speed \( v_{n} \) at \( x = L \) based on the relationship between speed and acceleration. This relationship can be formulated as:
\begin{equation}
\begin{aligned}
\tilde{v}_n = G(v_n) &\equiv 
\begin{cases} 
    \sqrt{v_n^2 + 2a_0 d(L) \Delta n} & \text{if } v_n < \bar{v}, \\ 
    u & \text{if } v_n \ge \bar{v},
\end{cases}
\\
\quad &\text{where } \bar{v} \equiv \sqrt{u^2 - 2a_0 d(L) \Delta n}.
\end{aligned}
\label{BAEQ}
\end{equation}

Combining the above two relationships, the reduced model is defined as a composite function of Eqs.~\eqref{EQBA} and~\eqref{BAEQ}, yielding:
 {
\begin{align}
& v_{n + \Delta n} = \left(\alpha\Delta n + \frac{1 + \gamma\Delta n}{\sqrt{v_{n}^{2} + \beta \Delta n}}\right)^{-1}, 
\quad n=1,2,3,\ldots \label{IFS}\\
&\text{where} \quad \alpha = \tau_{x}(L) 
= \frac{l_{1} - l_{2}}{Ll_{2}}\tau(L), \ \beta = 2a_0d(L),  \notag\\
& \mspace{65mu}\gamma = d_{x}(L) 
= \frac{l_{1} - l_{2}}{Ll_{2}}d(L). \notag
\end{align}}
This model allows for the iterative computation of vehicle speed for any vehicle, given the initial speed at the downstream end of the bottleneck section \( x=L \). 
The traffic flow rate can also be computed from the fundamental diagram as \( q(t,L) = \frac{v_n}{s(t,n)} \).
 {
It can also incorporate lane-changing intensity by replacing $l(x)$ with the effective number of lanes $\frac{l(x)}{1 + \eta(x)}$.
More specifically, one can replace \( l(x) \) in \( d(x) = \frac{1}{l(x)\kappa}, \ \tau(x) = \frac{1}{l(x)w\kappa} \) with \( \frac{l(x)}{1 + \eta(x)} \), and set \( \alpha \) and \( \gamma \) accordingly.
}

The green line in Fig.~\ref{kaisekiIFS} shows the result obtained using this model. The parameter settings are the same as those in Section~\ref{continuum}. From the figure, we confirm that the results of the continuum model in Lagrangian coordinates and the reduced model are in agreement.

\subsection{Stability of the capacity drop phenomenon}\label{sec:FP}
The second order model~\eqref{BALWR} captures an essential characteristic of the capacity drop phenomenon, which is summarized in the form of the following theorem. 

\begin{theorem}
The mapping \( f: [0, \bar{v}] \to [0, \bar{v}] \), defined by \eqref{IFS}, has a unique fixed point and its fixed point is globally stable.
\end{theorem}

\begin{proof}
For any \( v_a, v_b \) such that \( v_a < v_b \), by the mean value theorem, there exists some \( v_c \in [0, \bar{v}] \) such that:
\begin{equation}
\begin{aligned}
&\frac{f(v_b)-f(v_a)}{v_b-v_a}\\
&= f'(v_c) \\
&\equiv \frac{v_c(1 + \gamma \Delta n)}{\left( \sqrt{v_c^2 + \beta \Delta n} \right)^3 \left( \alpha \Delta n + \frac{1 + \gamma \Delta n}{\sqrt{v_c^2 + \beta \Delta n}} \right)^2 } \\
&= \frac{v_c(1 + \gamma \Delta n)}{ \sqrt{v_c^2 + \beta \Delta n} \left(1 + \gamma \Delta n + \alpha \Delta n \sqrt{v_c^2 + \beta \Delta n} \right)^2 }.
\end{aligned}
\end{equation}
From Equation \eqref{IFS}, we have \( \alpha > 0, \beta > 0, \gamma > 0 \). Therefore, the following inequalities hold:
\begin{align}
& \frac{v_c}{\sqrt{v_c^2 + \beta \Delta n}} = \frac{1}{\sqrt{1 + \frac{\beta \Delta n}{v_c^2}}} < 1, \\
& \frac{1 + \gamma \Delta n}{\left(1 + \gamma \Delta n + \alpha \Delta n \sqrt{v_c^2 + \beta \Delta n} \right)^2} < 1.
\end{align}
Thus,
\begin{align}
\left| \frac{f(v_b) - f(v_a)}{v_b - v_a} \right| = |f'(v_c)| < 1.
\end{align}
This shows that \( f \) is a contraction mapping. Since \( [0, \bar{v}] \) is a complete metric space under the standard absolute value metric, the Banach fixed-point theorem implies that the mapping \( f \) has a unique fixed point:
\begin{align}
f(v^*) = v^*.
\end{align}
\end{proof}

This theorem implies that the second order model captures a key characteristic of the capacity drop phenomenon: the capacity drop stationary state is the most stable, and the system converges to this stable state from any initial condition.

\begin{figure*}[t]
\centering
\begin{minipage}{0.48\textwidth}
  \centering
  \includegraphics[width=\linewidth,clip]{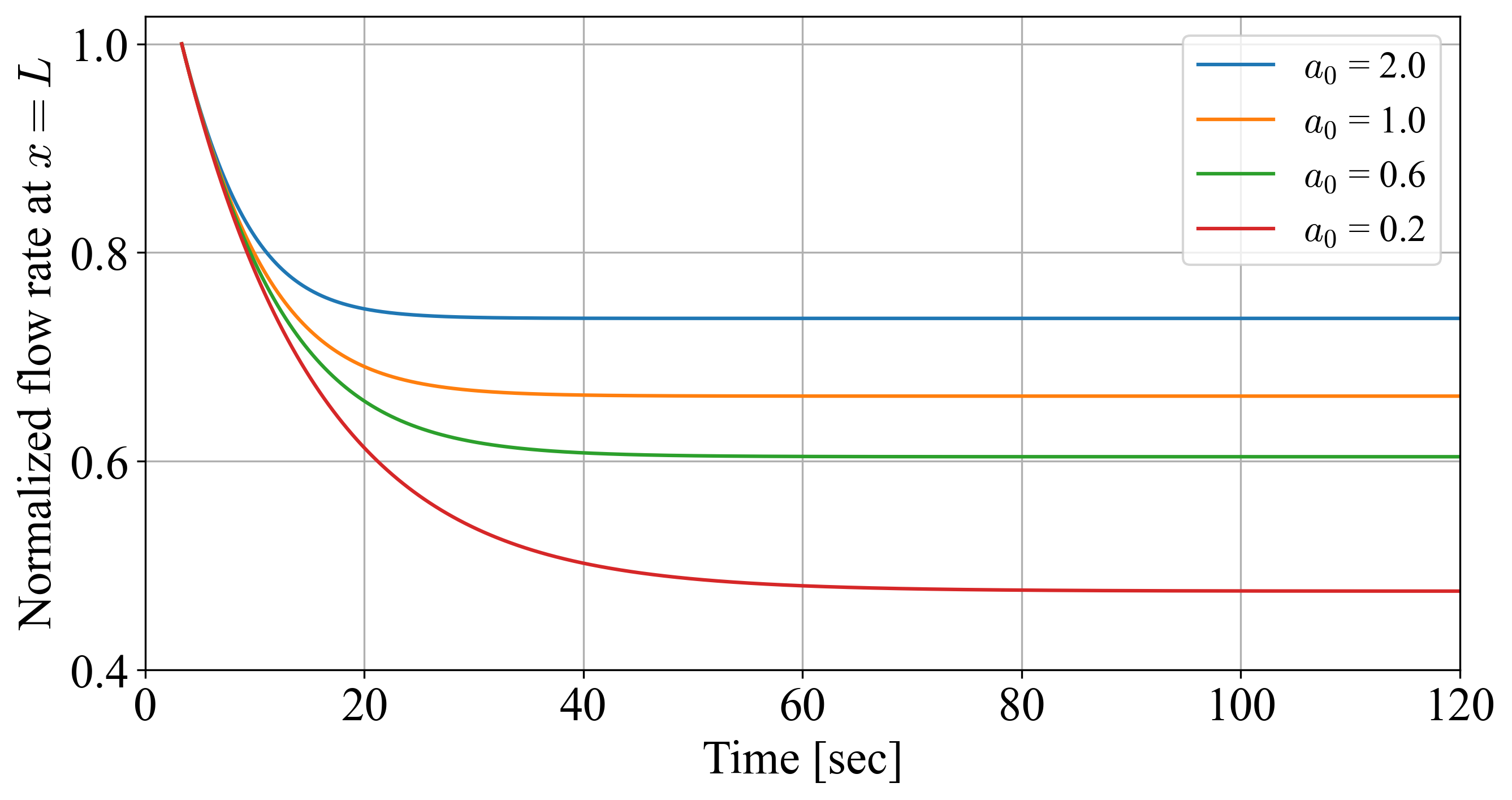}
  \caption{ {Relationship between acceleration and speed of convergence}}
  \label{kaiseki1}
\end{minipage}
\hfill
\begin{minipage}{0.48\textwidth}
  \centering
  \includegraphics[width=\linewidth,clip]{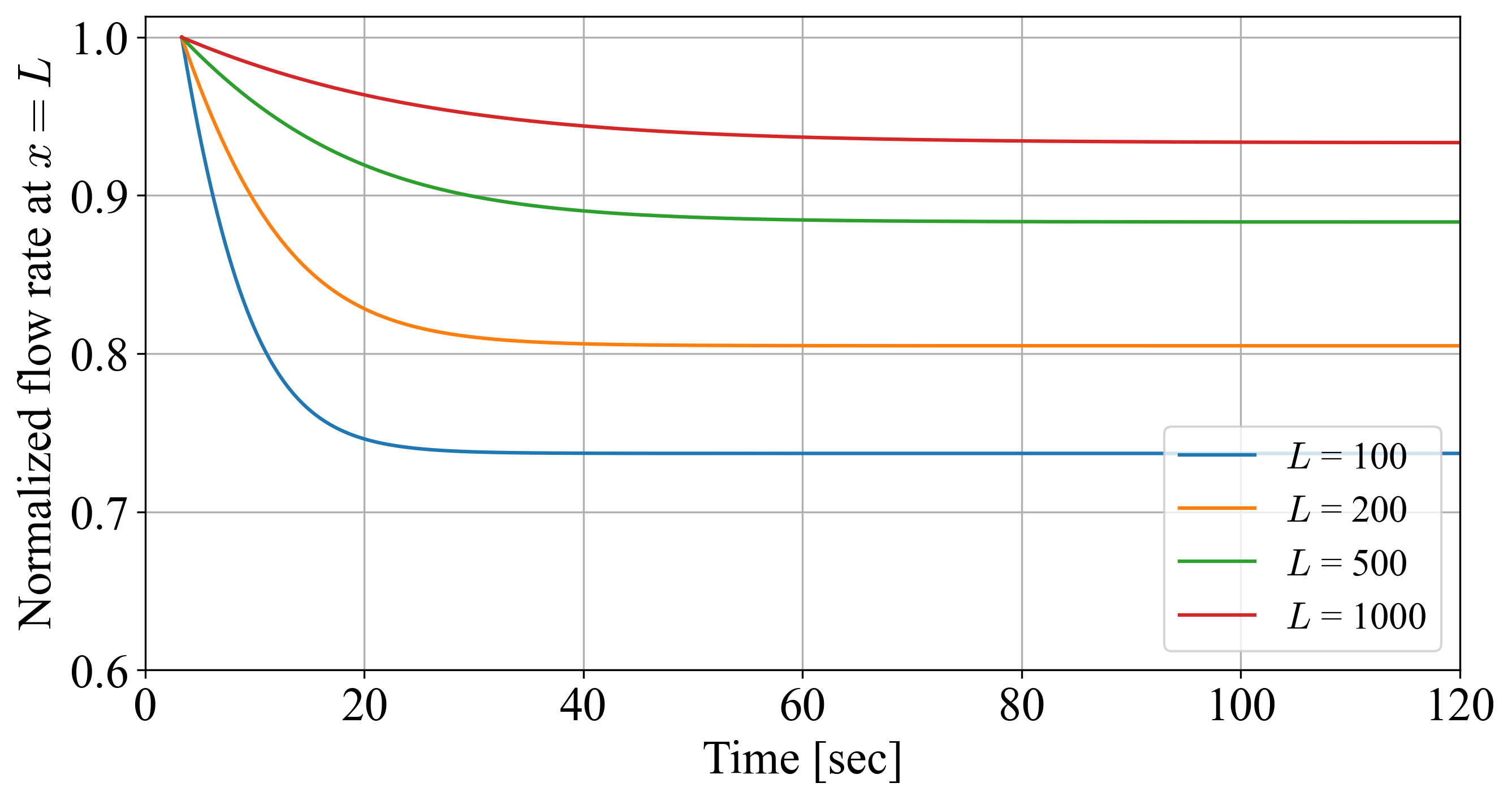}
  \caption{Relationship between bottleneck length and speed of convergence}
  \label{kaiseki2}
\end{minipage}
\end{figure*}

\begin{figure*}[t]
\centering
\begin{minipage}{0.48\textwidth}
  \centering
  \includegraphics[width=\linewidth,clip]{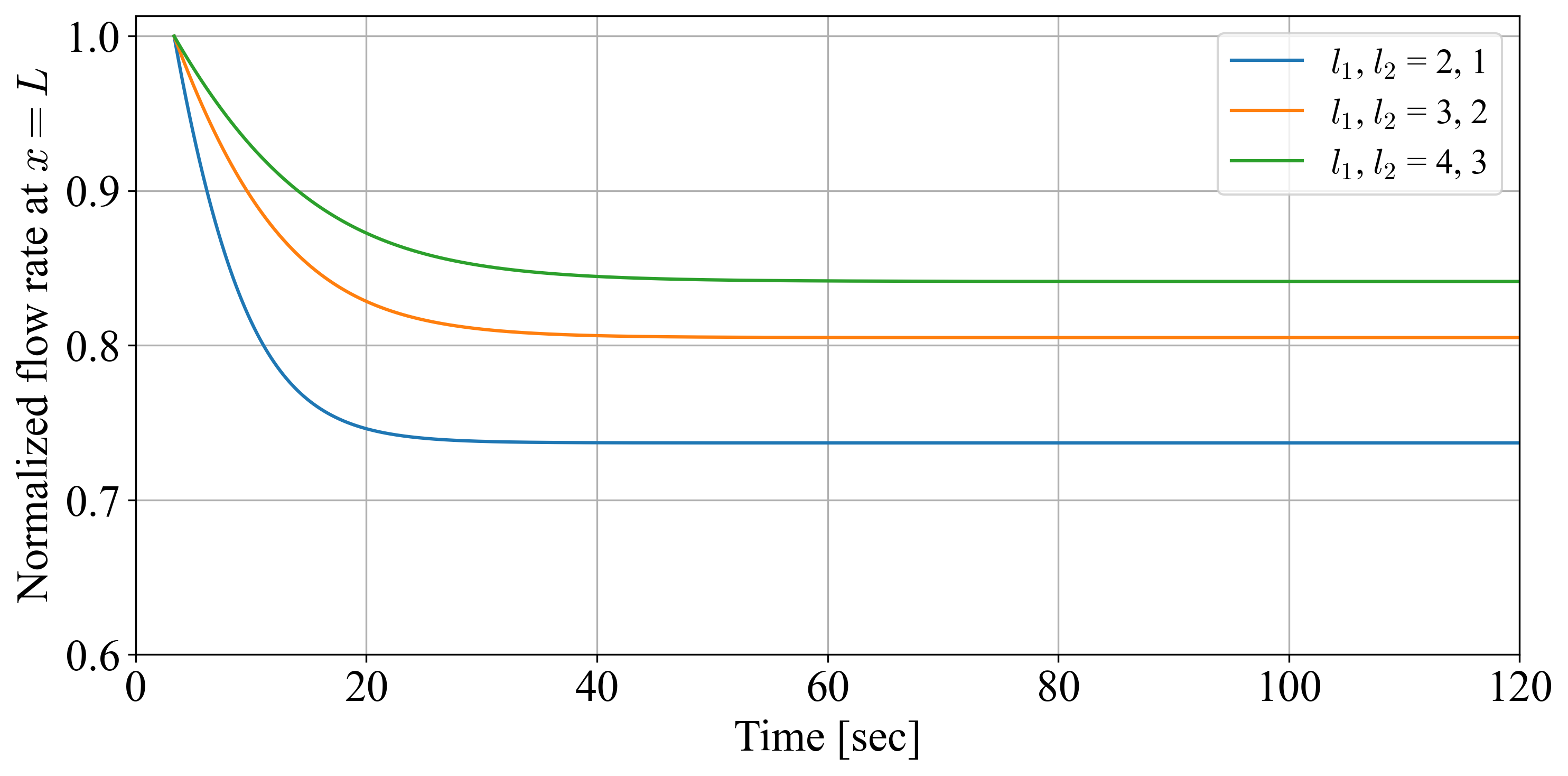}
  \caption{ {Relationship between the number of lanes and speed of convergence}}
  \label{kaiseki3}
\end{minipage}
\hfill
\begin{minipage}{0.48\textwidth}
  \centering
  \includegraphics[width=\linewidth,clip]{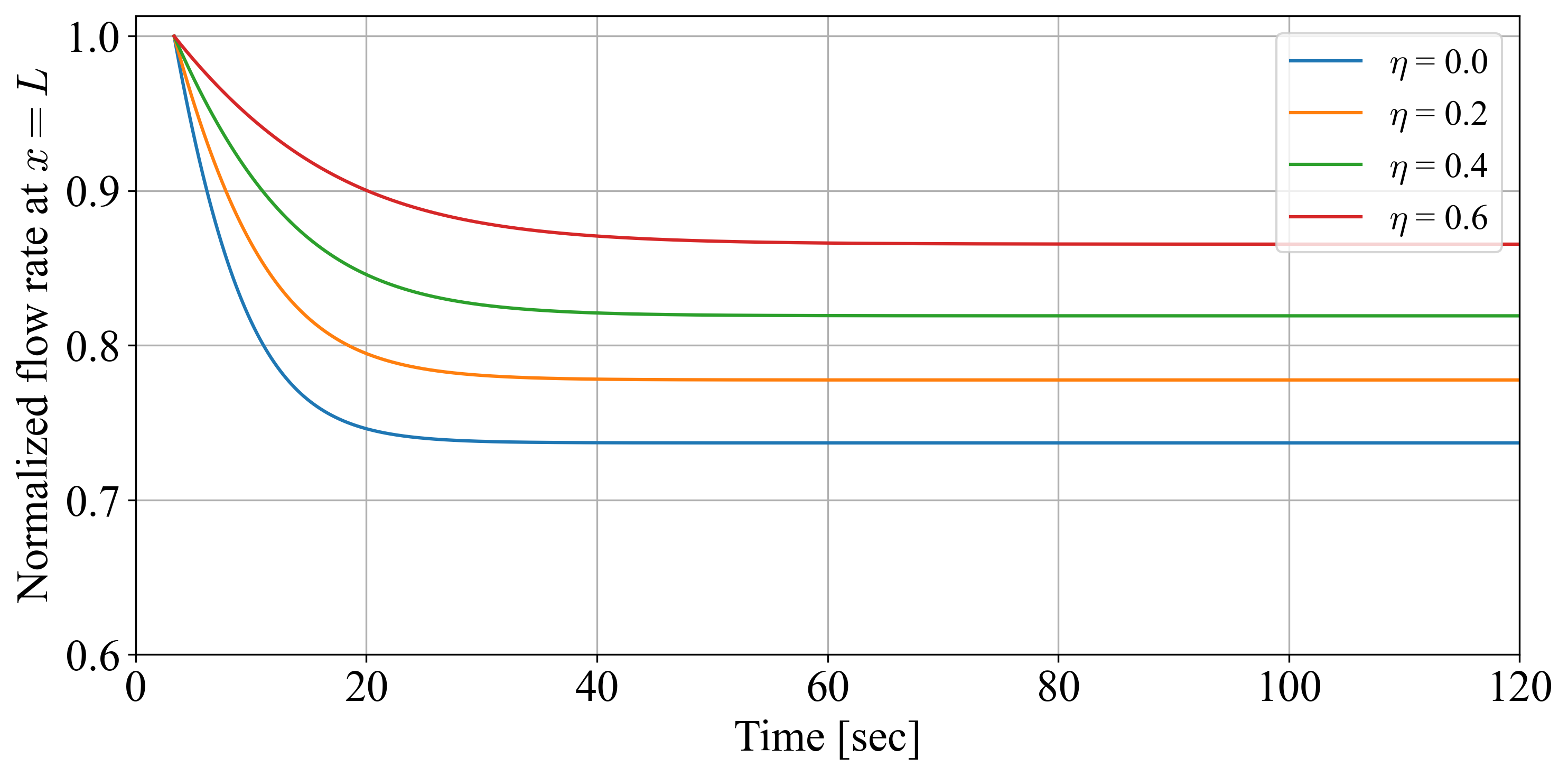}
  \caption{Relationship between lane-changing intensity and speed of convergence}
  \label{kaiseki4}
\end{minipage}\vspace{-3mm}
\end{figure*}

\subsection{Sensitivity Analysis}
In this section, we conduct a sensitivity analysis of the capacity drop phenomenon using the reduced model constructed in Section~\ref{sec:RM}.
Figures~\ref{kaiseki1}, \ref{kaiseki2}, \ref{kaiseki3}, and \ref{kaiseki4} show how the convergence speed of the capacity drop phenomenon and capacity drop ratio vary when changing each of the following parameters: acceleration \( a_0 \), bottleneck length \( L \), number of lanes \( l_1, l_2 \), and lane-changing intensity \( \eta \).
 {
Regarding lane-changing intensity, the number of lanes at the upstream end of the lane-drop section $l_{1}$ is replaced by $\frac{l_{1}}{1 + \eta}$, and the reduced model is computed.
}

From the results of the sensitivity analysis and the theoretical results discussed in Section~\ref{sec:FP}, 
we reveal a novel insight: at lane drop bottlenecks, once congestion occurs, the traffic flow immediately reaches the capacity drop stationary state.
This finding contrasts with that of Wada et al.~\cite{Wada}, who analyzed sag and tunnel bottlenecks and showed that stabilization requires a certain amount of time (e.g., approximately 20 minutes in their analysis based on the Kobotoke Tunnel).
 {
This difference can be attributed to the mechanisms of capacity reduction at the two types of bottlenecks. 
Specifically, at sag and tunnel sections, the bottleneck is represented by a decrease in the backward wave speed of the FD, which in Lagrangian coordinates corresponds to an increase in $\tau(x)$ within the bottleneck section (see Wada et al.~\cite{Wada}). 
In contrast, at lane-drop sections, the bottleneck is modeled by reducing the FD in a self-similar manner due to the decrease in the number of lanes (Fig.~\ref{FD}).
In Lagrangian coordinates, this change implies that not only $\tau(x)$ but also $d(x)$ increases simultaneously. 
This structural difference may account for the difference in stabilizing behavior.
}

More detailed results of the sensitivity analysis are as follows. The basic parameter settings are the same as in Section~\ref{continuum}. The convergence condition in model \eqref{IFS} is defined as $\lvert v_i - v^* \rvert < 10^{-2}$.

In Fig.~\ref{kaiseki1}, the convergence times for accelerations $a_0 =  {2.0},  {1.0},  {0.6},  {0.2}~\mathrm{m/s^2}$ are approximately $ {35.0}~\text{s}$, $ {44.9}~\text{s}$, $ {54.2}~\text{s}$ and $ {82.2}~\text{s}$, respectively. The corresponding capacity drop ratios, given by $(1 - \frac{C^-}{C(L)})$, are approximately $ {0.263}$, $ {0.337}$, $ {0.395}$ and $ {0.524}$.
These results indicate that a smaller acceleration leads to a longer convergence time and a larger capacity drop ratio.

In Fig.~\ref{kaiseki2}, the convergence times for bottleneck lengths $L = 100, 200, 500, 1000~\mathrm{m}$ are approximately 
$ {35.0}\,\text{s}$, $ {51.4}\,\text{s}$, $ {86.0}\,\text{s}$, and $ {124.3}\,\text{s}$, respectively. 
The corresponding capacity drop ratios are 
$ {0.263}$, $ {0.195}$, $ {0.117}$, and $ {0.067}$.
These results show that longer bottleneck lengths result in longer convergence times, but the capacity drop ratio becomes smaller.

In Fig.~\ref{kaiseki3}, the convergence times for lane configurations $(l_1, l_2) = (2,1), (3,2), (4,3)$ are approximately 
$ {35.0}\,\text{s}$, $ {51.4}\,\text{s}$, and $ {64.5}\,\text{s}$, respectively. 
The corresponding capacity drop ratios are 
$ {0.263}$, $ {0.195}$, and $ {0.158}$.
 {
A larger relative reduction in the number of lanes leads to a shorter convergence time but a larger capacity drop ratio.
}

\begin{figure*}[tb]
\centering
\includegraphics[width=0.9\textwidth,clip]{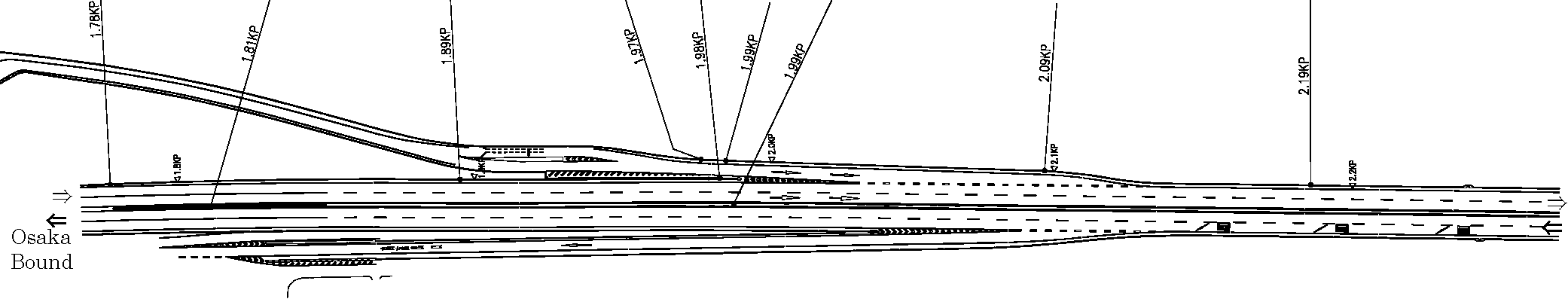}
\caption{Road section overview (from Zen Traffic Data \cite{Zen})}
\label{kukan}
\end{figure*}

\begin{figure*}[tb]
\centering
\includegraphics[width=0.9\textwidth,clip]{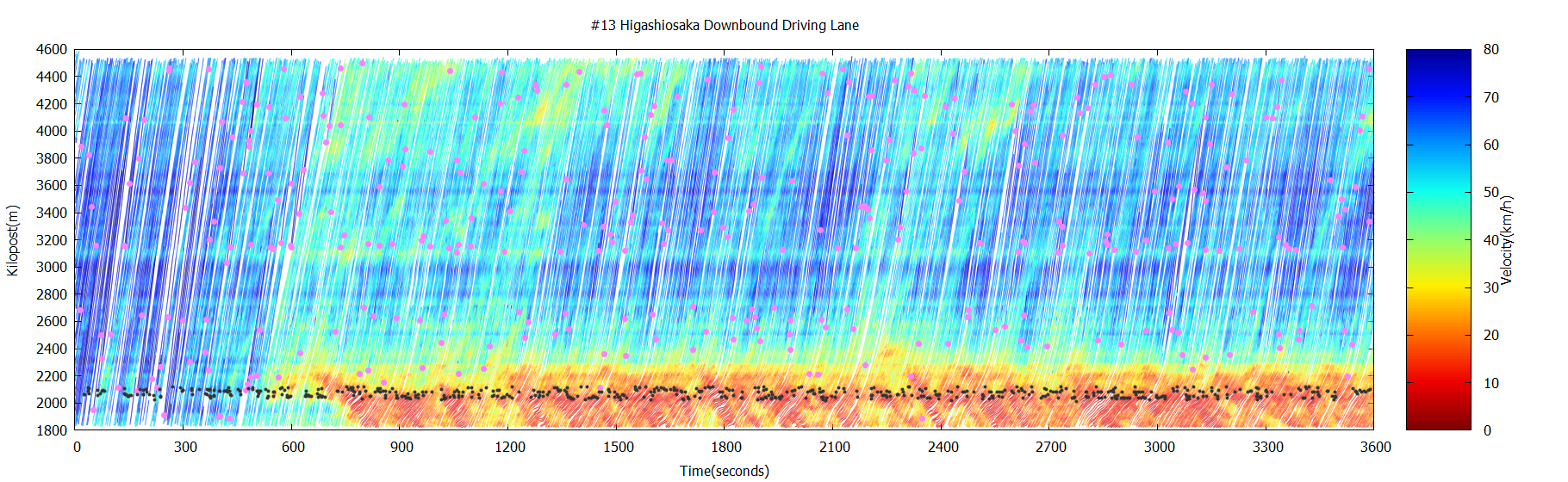}
\caption{Time-Space Diagram (from Zen Traffic Data \cite{Zen})}
\label{ts}
\end{figure*}

\begin{figure}[tb]
\centering
\includegraphics[width=0.48\textwidth,clip]{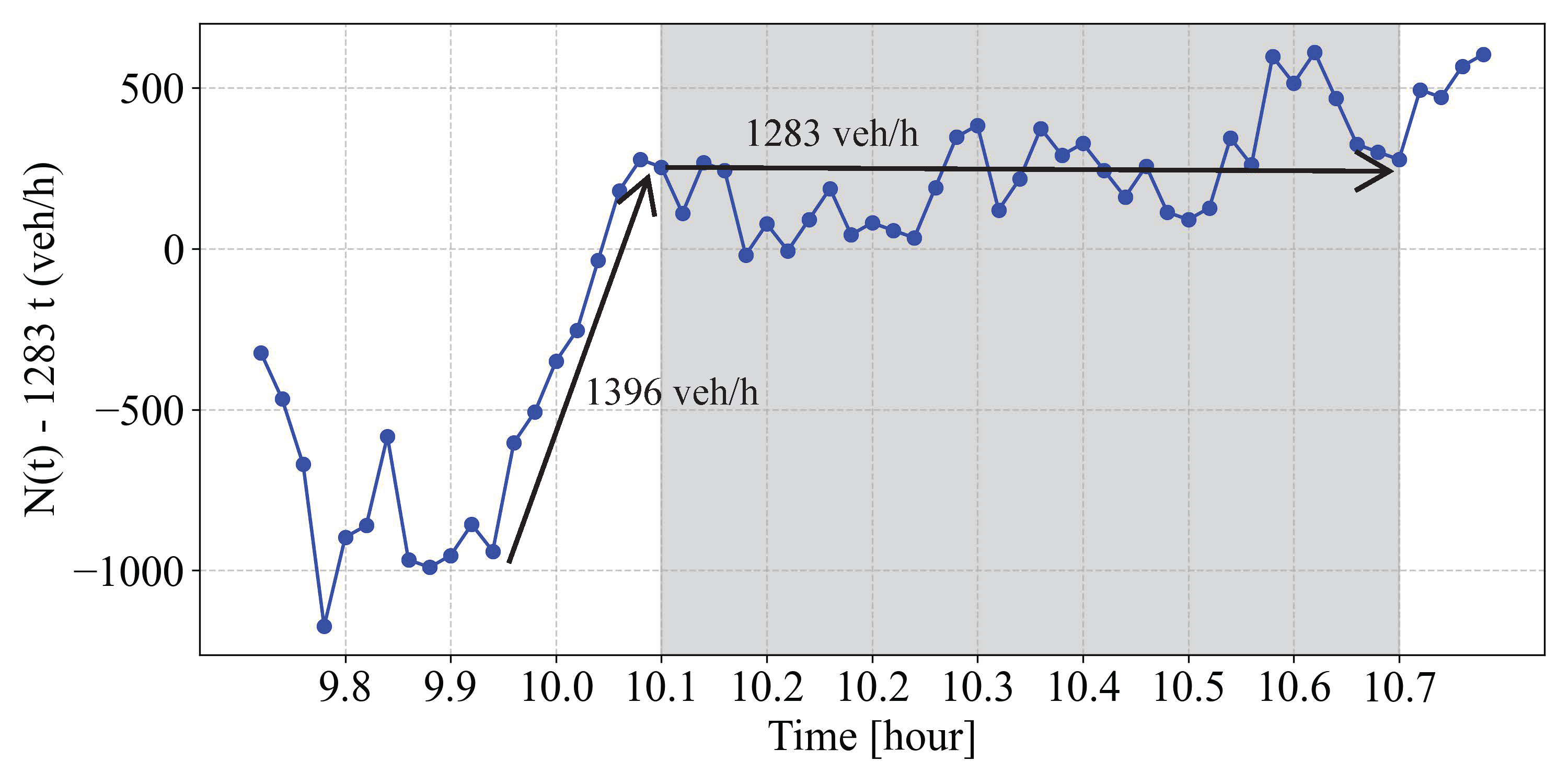}
\caption{Adjusted cumulative flow curve at 2.25 KP with 1283 veh/h baseline removed}
\label{cum}
\end{figure}

In Fig.~\ref{kaiseki4}, the convergence times for different lane-changing intensities $\eta = 0.0, 0.2, 0.4, 0.6$ are approximately 
$ {35.0}\,\text{s}$, $ {43.7}\,\text{s}$, $ {56.0}\,\text{s}$, and $ {75.9}\,\text{s}$, respectively. 
The corresponding capacity drop ratios are 
$ {0.263}$, $ {0.222}$, $ {0.181}$, and $ {0.134}$.
These results indicate that higher lane-changing intensity leads to longer convergence time.
 {
In addition, higher lane-changing intensity results in a smaller capacity drop ratio. This occurs because, in the present sensitivity analysis, the effective number of lanes is reduced only at the upstream end, which reduces the relative reduction in the number of lanes.
}

\section{Calibration and Validation}\label{sec:empirical}
In this section, we calibrate the extended model developed in the previous sections. Based on this calibrated results, we validate whether it can explain real-world congestion phenomena.

\subsection{Validation Framework}

The calibration consists of the following three steps, given the QDF \(q^*\), backward wave speed \( w \), and free-flow speed \( u \):  
(i) determine the bottleneck section,  
(ii) estimate the jam density of all lanes \( l(x)\kappa \) at each location, and  
(iii) estimate the acceleration parameter \( a_0 \).

In step (i), the upstream end of the bottleneck section is first determined by identifying the point at which speed recovery begins, based on the empirical speed recovery profile observed during congestion.
Next, if the recovery profile shows a shape transition from convex to concave, the inflection point—where vehicle behavior transitions from equilibrium congested states to BA states—is defined as the downstream end of the bottleneck section. 
This criterion is based on the characteristics clarified in Section~\ref{continuum}, particularly Figs.~\ref{sim} and~\ref{spdp}.  

In step (ii), the jam density of all lanes \( l(x)\kappa \) within the bottleneck section is estimated from the space-mean speed using the following equation:
\begin{align}
l(x)\kappa = q^* \left( \frac{1}{v^*(x)} + \frac{1}{w} \right).
\label{lxkj}
\end{align}
 {
This estimate can be interpreted as including the  lane-changing intensity $\eta(x)$.
For simplicity, however, we do not explicitly consider it here.
In any case, as shown below, the capacity can be estimated regardless of whether lane-changing intensity is included in the estimate. 
}

Finally, in step (iii), the acceleration parameter \( a_0 \) is estimated using the parameters obtained above:
\begin{align}
a^*(L) = \frac{dl(L)\kappa (q^*)^2 w^3}{\left(l(L)\kappa w - q^* \right)^3} = A(L, v^*(L)).
\end{align}
As the above procedure shows, once the bottleneck section is defined, the jam density \( l(x)\kappa \) and the acceleration parameter \( a_0 \) are uniquely determined in step (ii),(iii). 
Therefore, after estimating \( a_0 \) in step (iii), the model's reproducibility with respect to the observed speed recovery profile is evaluated. 
If necessary, the bottleneck section is adjusted by returning to step (i).

\begin{figure*}[t]
\centering
\includegraphics[width=0.6\textwidth,clip]{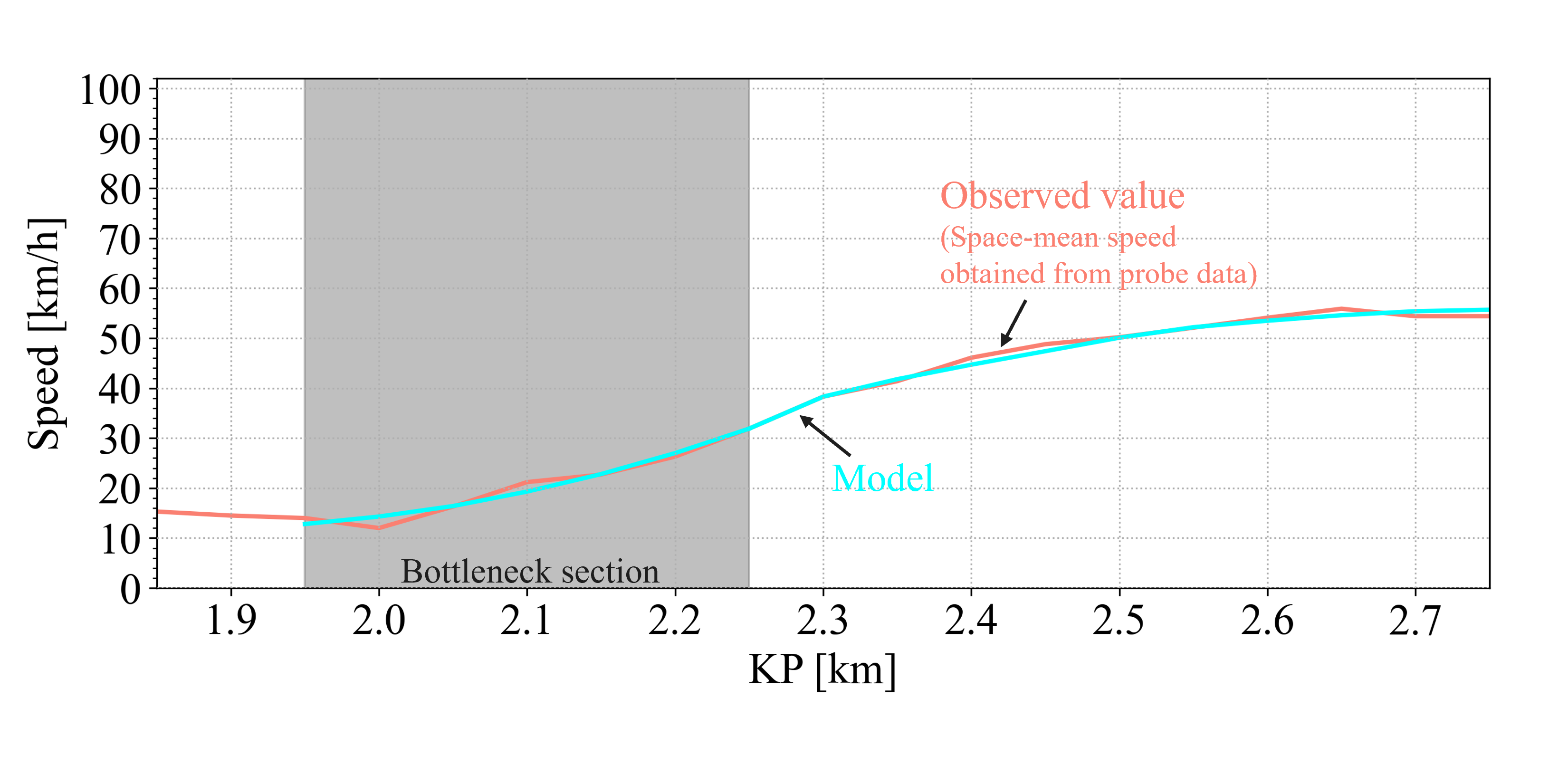}
\vspace{-4mm}
\caption{Speed recovery profile}
\label{SPD}
\end{figure*}

\begin{figure*}[t]
\centering
\includegraphics[width=0.6\textwidth,clip]{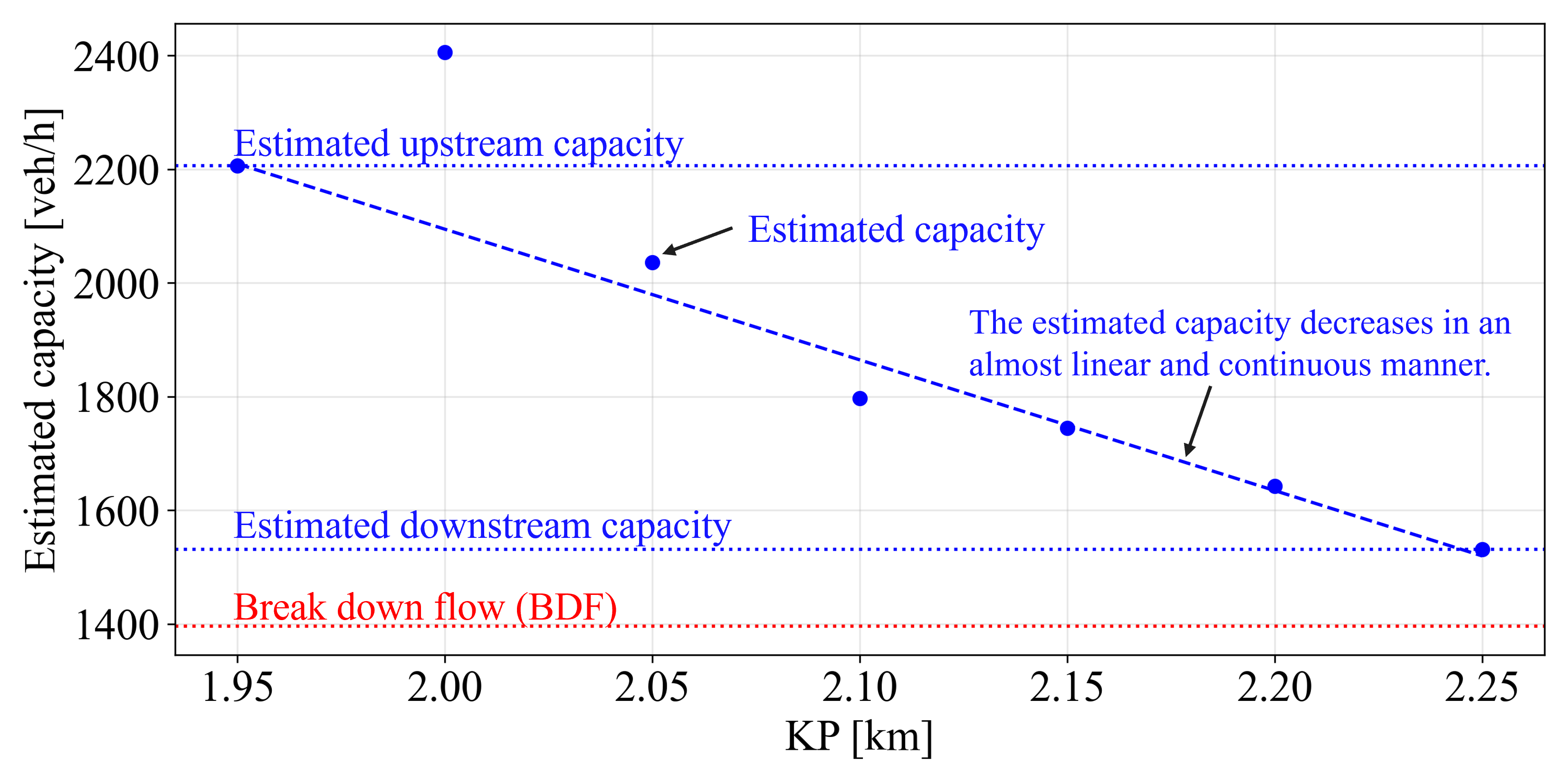}
\caption{Estimated capacity}
\label{estcx}
\end{figure*}

\subsection{Validation Results}

The empirical validation uses the Zen Traffic Data~\cite{Zen} on the mainline lane of the Hanshin Expressway Route 13 Higashi-Osaka Line, near the Morinomiya merging section (merge point: 2.1 KP). 
This is one of the major route passing through the city of Osaka and is one of the busiest routes on the Hanshin Expressway, with a daily traffic volume of 40,000 vehicles.
An overview of the analyzed merging section is shown in Fig.~\ref{kukan}, with the actual merging section located around 2.0--2.1 KP. 
In this study, we consider the merging lane and the leftmost driving lane as the target lanes.
Fig.~\ref{ts} displays the time-space diagram of the driving lane between 9:44 and 10:44 on the observation day. 
Fig.~\ref{cum} is an oblique cumulative plot at 2.25 KP, in which a constant trend of 1283 vehicles per hour has been subtracted to highlight the changes in the cumulative curve.
The shaded region corresponds to the duration of the stationary state.

From Fig.~\ref{ts}, the free-flow speed as \( u = 60~\mathrm{km/h} \), and the backward wave speed is estimated as \( w = 16.8~\mathrm{km/h} \). 
 {
Specifically, for \(u\), we roughly estimate it from the speed range in the free-flow part of Fig.~\ref{ts} (0--300 s), while \(w\) is obtained by averaging the upstream propagation speeds of several deceleration shock waves observed in the congested part (900--3600 s).
}
Furthermore, from Fig.~\ref{cum}, the breakdown flow (BDF) is estimated at \( 1396~\mathrm{veh/h} \), and the QDF at \( 1283~\mathrm{veh/h} \), indicating a capacity drop of approximately 8\%.
The subsequent analysis adopts these values as given.

First, we present the qualitative validation results. 
As shown in Fig.~\ref{cum}, the linear trend of the traffic flow changes rapidly and then quickly stabilizes once congestion occurs.
This observation  {is consistent with} the theoretical result presented in Section \ref{sec:Theory}, which revealed that the traffic flow immediately reaches the capacity drop stationary state after congestion occurs.
The estimated bottleneck section spans from 1.95 KP to 2.25 KP. Given that the actual geometric merging area lies around 2.0--2.1 KP, this estimation appears reasonable.

Next, we present the quantitative validation results using the framework described in the previous section
Fig.~\ref{SPD} compares the observed speed recovery profile with the model output. The model successfully reproduces the empirical speed profile, with a mean squared error of 0.89, indicating good agreement. 
Fig.~\ref{estcx} plots the estimated traffic capacity at each location. The capacity values are computed from the jam density using Eq.~\eqref{Cx} in step (ii). While there is a slight upward deviation at the 2.0 KP location, the overall trend of the estimated capacity is almost linear. 
This suggests consistency between the model assumptions and real-world traffic behavior.

On the other hand, when examining the relationship between the BDF and the estimated capacity, it is observed that the BDF is lower than the estimated bottleneck capacity. This suggests that the occurrence of congestion at the merging section cannot be explained by the simple mechanism of traffic volume exceeding capacity. More specifically, two mechanisms can be considered for congestion occurrence at the merging section. 
First, lane changing vehicle at the merge area plays a role. Since lane changes frequently occur in the merging zone, vehicle interactions between merging and mainline vehicles are likely to cause disruptions. Therefore, congestion may have occurred at a traffic volume lower than the actual capacity.
Second, merging vehicles may act as a moving bottleneck until they accelerate sufficiently after merging. In this case, the actual capacity might have been reduced below the capacity estimated by the theory.
In summary, although the validation results leave room for further discussion regarding the onset of congestion, they suggest that the proposed model has applicability to real-world congestion phenomena.

\section{Conclusion}\label{sec:conclusion}

This study first presented a second order model that endogenously describes the capacity drop phenomenon at lane-drop bottlenecks, following the framework proposed by Jin~\cite{Jin2017}. It was confirmed that the model reproduces the same stationary states as Jin's model and successfully captures the transitional process from congestion onset to the capacity drop stationary state.

Next, we proposed a one dimensional version of the second order model, termed the reduced model, and conducted theoretical analysis and sensitivity studies. The extended model was found to exhibit the following key properties:  
(i) the capacity drop stationary state is the most stable, and the system converges to this state from any initial condition;  
(ii) once congestion occurs, the traffic flow immediately reaches the capacity drop stationary state.  
Property (i) indicates that the theory captures a critical characteristic of capacity drop. 
Property (ii) offers a novel insight, suggesting that the temporal stabilization behavior of capacity drop differs significantly between lane-drop sections and other bottlenecks such as sags or tunnels, where stabilization may take longer.

 {
Finally, the theory was validated using one set of empirical traffic data. The results indicated that the rapid stabilization of the congestion pattern, as suggested by the model, can also be observed in real traffic, and that the estimation of the bottleneck section is reasonable. While this validation provides supportive evidence, it is only a first step toward a more comprehensive understanding of the phenomena based on the proposed theory.
}

Future research will validate the proposed framework across a broader range of lane-drop bottlenecks to assess the generality of the findings. 
We also plan to examine whether the theoretical properties derived here hold universally across different merging bottlenecks. 
Based on these insights, we aim to inform traffic control strategies and roadway design practices to mitigate the adverse effects of capacity drop.

\appendix
\section{ {Derivation of Model \eqref{overallCF}}}\label{sec:derivation}
 {Since the derivation of the speed–spacing FD~\eqref{CF2} from the flow–density FD~\eqref{sec:FD} is straightforward 
and Eq.~\eqref{CF3} holds by definition, only the derivation of Eq.~\eqref{CF} is presented here.}

 {Using new state variables \eqref{vtn}, \eqref{stn} and $a(t, n) = v_t + v v_x$, Eq.~\eqref{eq:velocity} can be rewritten as follows:  
\begin{align}
a(t, n) = \min \left\{ A(v, x), \frac{V(s, x)-v(t, n)}{\epsilon} \right\}\label{atnr}.
\end{align}
By setting $\epsilon = \Delta t$ and using Eq.~\eqref{a} , Eq.~\eqref{atnr} can be further approximated as 
\begin{align}
& \frac{v(t + \Delta t, n)-v(t, n)}{\Delta t} = \min \left\{ A(v, x), \frac{V(s, x)-v(t, n)}{\Delta t} \right\}\notag\\
& \Leftrightarrow \quad v(t + \Delta t, n) = \min\{V(s, x), v(t, n) + A(v, x)\Delta t\}. 
\label{vvtr}
\end{align}
Finally, by substituting Eq.~\eqref{vvtr} into the approximated relation \eqref{v},  $X(t + \Delta t, n) = X(t, n) + v(t + \Delta t, n)\Delta t$, we obtain the continuum model in Lagrangian coordinates \eqref{CF}.
Note that the conservation law (Eq.~\eqref{eq:conservation}) is automatically satisfied under the mild condition (i.e., continuity of $X(t, n)$), and therefore does not need to be explicitly considered.}


\section*{Declarations}

\subsection*{\textbf{Ethics approval and consent to participate}}
Not applicable.

\subsection*{\textbf{Consent for publication}}
Not applicable.

\subsection*{\textbf{Availability of data and materials}}
The data that support the findings of this study are available in Zen Traffic Data at https://zen-traffic-data.net, reference. 
These data are available for research and development organizations to contribute to the development of basic research, technology and services that will make road traffic more safe, secure and comfortable for the next generation.

\subsection*{\textbf{Competing interests}}
The authors declare that they have no competing interests.

\subsection*{\textbf{Funding}}
JSPS Grant-in-aid(KAKENHI) \#23K26218.

\subsection*{\textbf{Author's contributions}}
\textbf{FH:} Methodology, Software, Formal analysis, Validation, Writing - origianl draft, 
\textbf{KW:} Conceptualization, Methodology, Validation, Writing – review \& editing, Supervision, Funding acquisition.

\subsection*{\textbf{Acknowledgements}}
The authors express their gratitude to two anonymous referees for their careful reading of the manuscript and useful suggestions. 
 This work was partially supported by JSPS Grant-in-aid (KAKENHI) \#23K26218. The vehicle trajectory data used in this study was provided by Hanshin Expressway Co.


%
%



\end{document}